\documentclass[11pt,3p]{elsarticle}

\usepackage{fixltx2e,amsmath,graphicx,amssymb,breakurl,amsthm}
\usepackage[mathlines,displaymath]{lineno}

\newcommand{\Suff}{\textit{Suff}}
\newcommand{\Pref}{\textit{Pref}}
\newcommand{\Fact}{\textit{Fact}}

\renewcommand{\alph}{\textit{alph}}

\renewcommand{\epsilon}{\varepsilon}

\newtheorem{theorem}{Theorem}[section]
\newtheorem{proposition}[theorem]{Proposition}
\newtheorem{lemma}[theorem]{Lemma}
\newtheorem{corollary}[theorem]{Corollary}
\theoremstyle{definition}
\newtheorem{definition}[theorem]{Definition}
\newtheorem{example}[theorem]{Example}

\newtheorem{remark}[theorem]{Remark}

\begin{document}

\begin{frontmatter}

\title{\textbf{Enumeration and Structure of Trapezoidal Words}\tnoteref{note1}}
\tnotetext[note1]{Part of the results in this paper were presented by the third author at \emph{WORDS 2011, 8th International Conference on Words}, 12-16 September 2011, Prague, Czech Republic \cite{Fi11}.}

\author{Michelangelo Bucci}
\ead{michelangelo.bucci@utu.fi}
\address{Department of Mathematics, University of Turku\\
FI-20014 Turku, Finland}
 
\author{Alessandro De Luca}
 \ead{alessandro.deluca@unina.it}
 \address{Dipartimento di Scienze Fisiche\\ Universit\`a degli Studi di Napoli Federico II\\ Via Cintia, Monte S.~Angelo, I-80126 Napoli, Italy}
 
\author{Gabriele Fici\corref{cor1}}
\ead{gabriele.fici@unice.fr}
\address{Laboratoire d'Informatique, Signaux et Syst\`emes de Sophia-Antipolis \\ CNRS \& Universit\'e de Nice-Sophia Antipolis\\ 2000, Route des Lucioles - 06903 Sophia Antipolis cedex, France}

\cortext[cor1]{Corresponding author.}

\journal{Theoretical Computer Science}

\begin{abstract}
Trapezoidal words  are words having at most $n+1$ distinct factors of length $n$ for every $n\ge 0$. They therefore encompass finite Sturmian words. We give combinatorial characterizations of trapezoidal words and exhibit a formula for their enumeration. We then separate trapezoidal words into two disjoint classes: open and closed. A trapezoidal word is closed if it has a factor that occurs only as a prefix and as a suffix; otherwise it is open. We investigate open and closed trapezoidal words, in relation with their special factors. We prove that Sturmian palindromes are closed trapezoidal words and that a closed trapezoidal word is a Sturmian palindrome if and only if its longest repeated prefix is a palindrome. We also define a new class of words, \emph{semicentral words}, and show that they are characterized by the property that they can be written as $uxyu$, for a central word $u$ and two different letters $x,y$. Finally, we investigate the prefixes of the Fibonacci word with respect to the property of being open or closed trapezoidal words, and show that the sequence of open and closed prefixes of the Fibonacci word follows the Fibonacci sequence.
\end{abstract}

\begin{keyword}
Sturmian words, trapezoidal words, rich words, closed words, special factors, palindromes, enumerative formula. 
\end{keyword}

\end{frontmatter}


\section{Introduction}\label{sec:intro}

In combinatorics on words, the most famous class of infinite words is certainly that of Sturmian words. They code digital straight lines in the discrete plane having irrational slope, and are characterized by the property of having exactly $n+1$ factors of length $n$, for every $n\ge 0$.

It is well known (see \cite{LothaireAlg}, Proposition 2.1.17) that a finite word $w$ is a factor of some Sturmian word if and only if $w$ is a binary \emph{balanced} word, that is, for every letter $a$ and every pair of factors of $w$  of the same length, $u$ and $v$, one has that $u$ and $v$ contain the same number of $a$'s up to one, i.e.,
\begin{equation}\label{eq:bal}
 ||u|_{a}-|v|_{a}|\le 1.
\end{equation}

Finite Sturmian words (finite factors of Sturmian words) have \emph{at most} $n+1$ factors of length $n$, for every $n\ge 0$. However, this property does not characterize them, as shown by the word $w=aaabab$, which is not Sturmian since the factors $aaa$ and $bab$ do not verify \eqref{eq:bal}.

The finite words defined by the property of having at most $n+1$ factors of length $n$, for every $n\ge 0$, are called \emph{trapezoidal words}. The name comes from the fact that the graph of the function counting the distinct factors of each length (usually called the factor complexity) defines, for trapezoidal words, an isosceles trapezoid.

Trapezoidal words were defined by de Luca \cite{Del99} by means of combinatorial parameters on the repeated factors of finite words. For a finite word $w$ over a finite alphabet $\Sigma$, let $K_{w}$ be the minimal length of an unrepeated suffix of $w$ and let $R_{w}$ be the minimal length for which $w$ does not contain right special factors, i.e., factors having occurrences followed by distinct letters. de Luca proved that for any word $w$ of length $|w|$ one always has
\begin{equation}\label{eq:RK}
 R_{w}+K_{w}\le|w|
\end{equation}
and studied the case in which the equality holds. This is in fact the case in which $w$ is trapezoidal, since one can prove (see Proposition \ref{prop:trap} later) that a word $w$ is trapezoidal if and only if
\begin{equation}\label{eq:trap}
 R_{w}+K_{w}=|w|.
\end{equation}
Then, de Luca proved that finite Sturmian words verify \eqref{eq:trap} and are therefore trapezoidal.

The non-Sturmian trapezoidal words were later characterized by D'Alessandro \cite{Dal02}. Since a non-Sturmian word is not balanced, it must contain a pair of factors having the same length and not verifying \eqref{eq:bal}. Such a pair is called a \emph{pathological pair}. Let $w$ be a non-Sturmian word and $(f,g)$ its pathological pair of minimal length. D'Alessandro proved that $w$ is trapezoidal if and only if $w=pq$ for some $p\in \Suff(\{\tilde{z}_{f}\}^{*})$ and $q\in \Pref(\{z_{g}\}^{*})$, where $\tilde{z}_{f}$ is the reversal of the fractional root $z_{f}$ of $f$ and $z_{g}$ is the fractional root of $g$. This characterization is quite involved and will be discussed in detail later.

The characterization by D'Alessandro allows us to give an enumerative formula for trapezoidal words. Indeed, it is known that the number of Sturmian words of length $n$ is equal to
\[S(n)=1+\sum_{i=1}^{n}(n-i+1)\phi(i)\]
(cf.~\cite{Mig91,Lip82}), where $\phi$ is the Euler totient function, i.e., $\phi(n)$ is the number of positive integers smaller than or equal to $n$ and coprime with $n$. In Section \ref{sec:formula} we prove that the number of non-Sturmian trapezoidal words of length $n$ is
\[T(n)=\sum_{i=0}^{\lfloor (n-4)/2\rfloor}2(n-2i-3)\phi(i+2),\]
and thus the number of trapezoidal words of length $n$ is $S(n)+T(n)$ (Theorem \ref{theor:enum}).


We then distinguish trapezoidal words into two distinct classes, accordingly with the definition below.

\begin{definition}\label{def:closed}
Let $w$ be a finite word over an alphabet $\Sigma$. We say that $w$ is \emph{closed} if it is empty or it has a factor (different from the word itself) occurring exactly twice in $w$, as a prefix and as a suffix, that is, with no internal occurrences. A word which is not closed is called \emph{open}.
\end{definition} 

\noindent For example, the words $aabbaa$ and $ababa$ are closed, whereas the word $aabbaaa$ is open. 

The notion of closed word is closely related to the concept of \emph{complete return} to a factor $u$ in a word $w$, as considered in \cite{GlJuWiZa09}. A complete return to $u$ in $w$ is any factor of $w$ having exactly two occurrences of $u$, one as a prefix and one as a suffix. Hence $w$ is closed if and only if it is a complete return to one of its factors; such a factor is clearly both the longest repeated prefix and the longest repeated suffix of $w$. 

The notion of closed word is also equivalent to that of \emph{periodic-like} word \cite{CaDel01a}. A word $w$ is periodic-like if its longest repeated prefix does not appear in $w$ followed by different letters.

We derive some properties of open and closed trapezoidal words in Section \ref{sec:openclosed}. We prove that the set of closed trapezoidal words is contained in that of Sturmian words (Proposition \ref{prop:cloStur}). 
We then characterize the special factors of closed trapezoidal words (Lemma \ref{lem:clo1}) by showing that every closed trapezoidal word $w$ contains a bispecial factor $u$ such that $u$ is a central word and the left (resp.\ right) special factors of $w$ are the prefixes (resp.\ suffixes) of $u$. We show that trapezoidal palindromes (hence, Sturmian palindromes) are closed words (Theorem \ref{theor:palclo}), and that a closed trapezoidal word is a Sturmian palindrome if and only if its longest repeated prefix is a palindrome (Proposition \ref{prop:palh}).

In Section \ref{sec:sesqui} we introduce \emph{semicentral words}, that are trapezoidal words in which the longest repeated prefix, the longest repeated suffix, the longest left special factor and the longest right special factor all coincide. We prove in Theorem \ref{theor:sesqui} that a trapezoidal word $w$ is semicentral if and only if it is of the form $w=uxyu$, for a central word $u$ and two different letters $x,y$.

We then study, in Section \ref{sec:fibo}, the prefixes of the Fibonacci infinite word $f$, and show that the sequence of the numbers of consecutive prefixes of $f$ that are closed (resp.~open) is the sequence of Fibonacci numbers (Theorem \ref{theor:fibo}).

The paper ends with conclusions and open problems in Section \ref{sec:conc}.

\section{Trapezoidal words}\label{sec:trap}


An \textit{alphabet}, denoted by $\Sigma$, is a finite set of symbols, called \emph{letters}. A \textit{word} over $\Sigma$ is a finite sequence of letters from $\Sigma$. The subset of the alphabet $\Sigma$ constituted by the letters appearing in $w$ is denoted by $\alph(w)$. We let $w_{i}$ denote the $i$-th letter of $w$.  Given a word $w=w_1w_2\cdots w_n$, with $w_i\in\Sigma$ for $1\leq i\leq n$, the nonnegative integer $n$ is the \emph{length} of $w$, denoted by $|w|$. The empty word has length zero and is denoted by $\varepsilon$. 
The set of all words over $\Sigma$ is denoted by $\Sigma^*$. The set of all words over $\Sigma$ having length $n$ is denoted by $\Sigma^n$. For a letter $a\in \Sigma$, $|w|_{a}$ denotes the number of $a$'s appearing in $w$.

The word $\tilde{w}=w_{n}w_{n-1}\cdots w_{1}$ is called the \emph{reversal} of $w$. A \emph{palindrome} is a word $w$ such that $\tilde{w}=w$. In particular, we assume $\tilde{\epsilon}=\epsilon$ and so the empty word is a palindrome.

A word $z$ is a \emph{factor} of a word $w$ if $w=uzv$ for some $u,v\in \Sigma^{*}$. In the special case $u = \varepsilon $ (resp.~$v = \varepsilon $), we call $z$ a prefix (resp.~a suffix) of $w$. An occurrence of a factor of a word $w$ is \emph{internal} if it is not a prefix nor a suffix of $w$. The set of prefixes, suffixes and factors of the word $w$ are denoted respectively by $\Pref(w)$, $\Suff(w)$ and $\Fact(w)$. A \emph{border} of the word $w$ is any word in $\Pref(w)\cap \Suff(w)$ different from $w$.

The \emph{factor complexity} of a word $w$ is the function defined by $f_{w}(n)=|\Fact(w)\cap \Sigma^n|$, for every $n\geq 0$. 

A factor $u$ of $w$ is \emph{left special} (resp.~\emph{right special}) in $w$ if there exist $a,b\in \Sigma$, $a\neq b$, such that $au,bu\in \Fact(w)$ (resp.~$ua,ub\in \Fact(w)$). A factor is \emph{bispecial} in $w$ if it is both left and right special. 

A \textit{period} for the word $w$ is a positive integer $p$, with $0<p\leq |w|$, such that
$w_{i}=w_{i+p}$ for every $i=1,\ldots ,|w|-p$. Since $|w|$ is always a period for $w$, we have that every non-empty word has at least one period. We can unambiguously define \textit{the} period of the word $w$ as the smallest of its periods. The period of $w$ is denoted $\pi_{w}$. The \emph{fractional root} $z_{w}$ of a word $w$ is its prefix whose length is equal to the 
period of $w$.  

A word $w$ is \emph{a power} if there exists a non-empty word $u$ and an integer $n>1$ such that $w=u^{n}$. A word which is not a power is called \emph{primitive}.

For example, let $w= aababba$. The left special factors of $w$ are $\epsilon$, $a$, $ab$, $b$ and $ba$. The right special factors of $w$ are $\epsilon$, $a$, $ab$ and $b$. Therefore, the bispecial factors of $w$ are $\epsilon$, $a$, $ab$ and $b$. The period of $w$ is $\pi_{w}=6$ and the fractional root of $w$ is $z_{w}=aababb$.

The following parameters were introduced by Aldo de Luca \cite{Del99}:

\begin{definition} Given a word $w\in \Sigma^{*}$, $H_w$ (resp.~$K_w$) denotes the minimal length of a prefix (resp.~a suffix) of $w$ which occurs only once in $w$.
\end{definition}

\begin{definition}
Given a word $w\in \Sigma^{*}$, $L_w$ (resp.~$R_w$) denote the minimal length for which there are no left special factors (resp.~no right special factors) of that length in $w$. 
\end{definition}

\begin{example}\label{ex:1}
Let $w=aaababa$. The longest left special factor of $w$ is $aba$, and it is also the longest repeated suffix of $w$; the longest right special factor of $w$ is $aa$, and it is also the longest repeated prefix of $w$. Thus, we have $L_{w}=4$, $K_{w}=4$, $R_{w}=3$ and $H_{w}=3$.
\end{example}

Notice that the values $H_{w},K_{w},L_{w}$ and $R_{w}$ are positive integers, unless for the words that are powers of a single letter, in which case $R_{w}=L_{w}=0$ and $K_{w}=H_{w}=|w|$. Moreover, one has $f_{w}(R_{w})=f_{w}(L_{w})$ and $\max\{R_{w},K_{w}\}=\max\{L_{w},H_{w}\}$ (cf.~\cite{Del99}, Corollary 4.1).

The following proposition is from de Luca (\cite{Del99}, Proposition 4.2).

\begin{proposition}\label{prop:del}
Let w be a word of length $|w|$ such that $|\alph(w)|> 1$ and set
$m_{w} = \min\{R_{w}, K_{w}\}$ and $M_{w} = \max\{R_{w}, K_{w}\}$. The factor complexity $f_{w}$ is strictly increasing
in the interval $[0,m_{w}]$, is nondecreasing in the interval $[m_{w},M_{w}]$ and strictly decreasing
in the interval $[M_{w}, |w|]$. Moreover, for $i$ in the interval $[M_{w},|w|]$, one has $f_{w}(i + 1) = f_{w}(i) - 1$. If $R_{w}<K_{w}$, then $f_{w}$ is constant in the interval $[m_{w},M_{w}]$.
\end{proposition}

Proposition \ref{prop:del} allows one to give the following definition.

\begin{definition}
A non-empty word $w$ is \emph{trapezoidal} if 

\begin{itemize}
 \item $f_{w}(i)=i+1$ \ \ for \ \ $0\le i \le m_{w}$,
 \item $f_{w}(i+1)=f_{w}(i)$ \ \ for \ \  $m_{w} \le i \le M_{w}-1$,
 \item $f_{w}(i+1)=f_{w}(i)-1$ \ \ for \ \  $M_{w} \le i \le|w|$.
\end{itemize}
\end{definition}

Trapezoidal words were considered for the first time by de Luca \cite{Del99}.  
The name \emph{trapezoidal} was given by D'Alessandro \cite{Dal02}. The choice of the name is motivated by the fact that for these words the graph of the complexity function defines an isosceles trapezoid (possibly degenerated in a triangle). 

\begin{figure}
\begin{center}
\includegraphics[height=60mm]{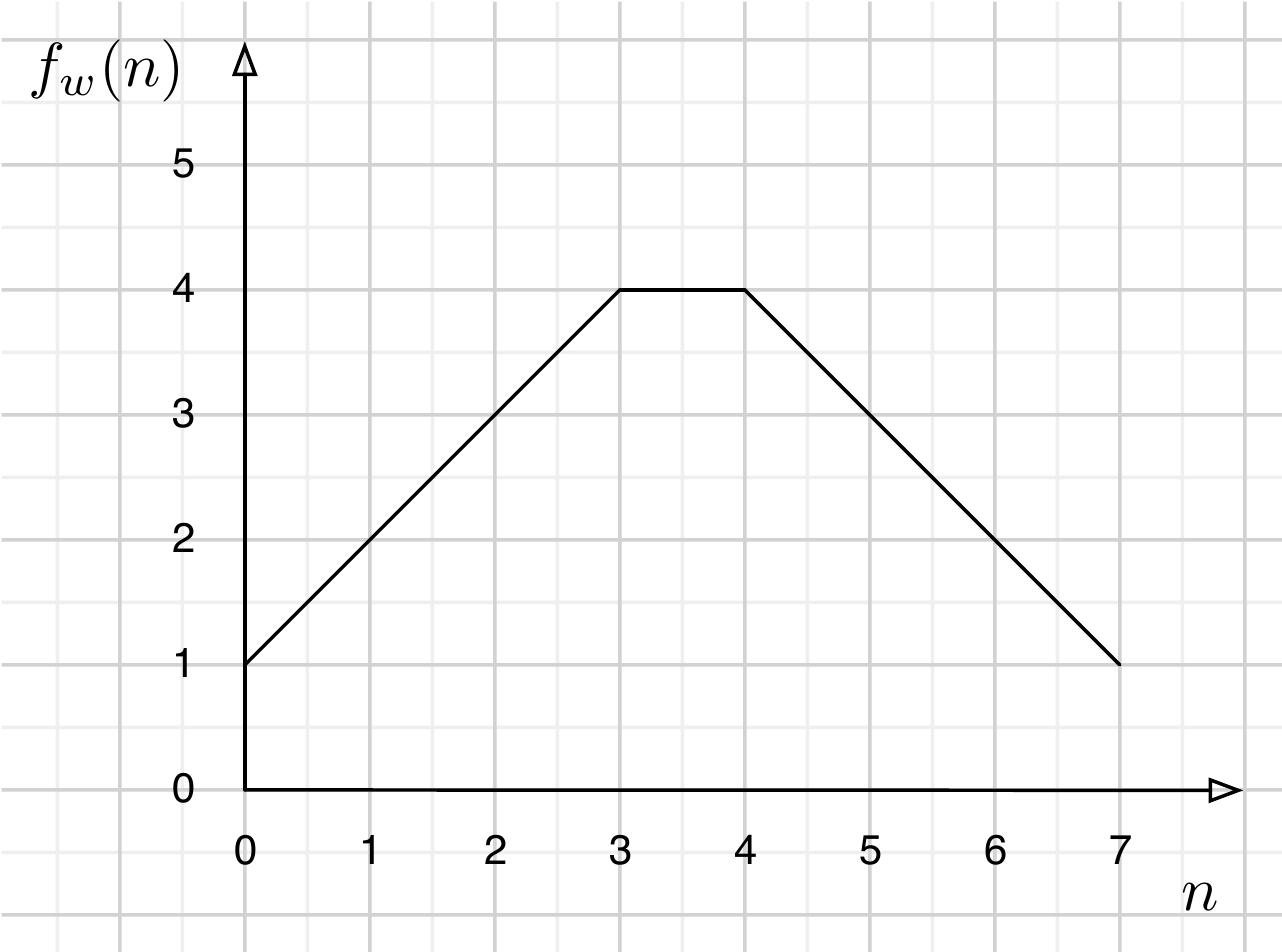}
\caption{The graph of the complexity function of the trapezoidal word $w=aaababa$. One has $\min \{R_{w},K_{w}\}=3$ and $\max \{R_{w},K_{w}\}=4$.}
\label{fig:graph}
\end{center}
\end{figure}

\begin{remark}
 It follows from the definition that a trapezoidal word is necessarily a binary word, since $f_{w}(1)$, the number of distinct letters appearing in $w$, must be $1$ or $2$. 
\end{remark}

\begin{example}
 The word $w=aaababa$ considered in Example \ref{ex:1} is trapezoidal. See \figurename~\ref{fig:graph}.
\end{example}

In the following proposition we gather some characterizations of trapezoidal words. 

\begin{proposition}\label{prop:trap}
Let $w$ be a word. The following conditions are equivalent:
 
\begin{enumerate}
\item $w$ is trapezoidal;
 \item $|w|=L_{w}+H_{w}$;
 \item $|w|=R_{w}+K_{w}$;
 \item $w$ is binary and has at most one left special factor of length $n$ for every $n\geq 0$;
 \item $w$ is binary and has at most one right special factor of length $n$ for every $n\geq 0$;
 \item $w$ has at most $n+1$ distinct factors of length $n$ for every $n\geq 0$;
 \item $|f_{w}(n+1)-f_{w}(n)|\leq 1$ for every $n\geq 0$.
\end{enumerate}
\end{proposition}

\begin{proof}
The equivalence of $(1)$, $(2)$, $(3)$, $(4)$ and $(5)$ is in \cite{Del99} and \cite{Dal02}. The equivalence of $(5)$ and $(6)$ follows from elementary considerations on the factor complexity of binary words. Indeed, it is easy to see the number of distinct factors of length $n+1$ of a binary word $w$ is at most equal to the number of distinct factors of length $n$ plus the number of right special factors of length $n$. 
The equivalence of $(7)$ and $(1)$ comes directly from the definitions and from Proposition \ref{prop:del}. 
\end{proof}

 From now on, we fix the alphabet $\Sigma=\{a,b\}$.

Recall that a word $w\in\Sigma^{*}$ is \emph{Sturmian} if  and only if it is balanced, i.e., verifies \eqref{eq:bal}. The following proposition is from de Luca (\cite{Del99}, Proposition 7.1).

\begin{proposition}\label{prop:sturmtrap}
Every Sturmian word is trapezoidal. 
\end{proposition}

The inclusion in Proposition \ref{prop:sturmtrap} is strict, since there exist trapezoidal words that are not Sturmian, e.g. the word $w=aaababa$ considered in Example \ref{ex:1}. 

Recall that a word $w$ is \emph{rich}~\cite{GlJuWiZa09} (or \emph{full}~\cite{BrHaNiRe04})  if it has $|w|+1$ distinct palindromic factors, that is the maximum number of distinct palindromic factors a word can contain. The following proposition is from de Luca, Glen and Zamboni  (\cite{DelGlZa08}, Proposition 2).

\begin{proposition}\label{prop:traprich}
Every trapezoidal word is rich.
\end{proposition}

Again, the inclusion in Proposition \ref{prop:traprich} is strict, since there exist binary rich words that are not trapezoidal, e.g. the word $w=aabbaa$.


We now discuss the case of non-Sturmian trapezoidal words. We first recall the definition and some properties of \emph{central words} \cite{DelMi94}.

\begin{definition}
A word $w\in\Sigma^{*}$ is a central word if and only if it has two coprime periods $p$ and $q$ and length equal to $p+q-2$.
\end{definition}

\begin{proposition}[\cite{ilieplan,CarDel05}]
\label{prop:ipcent}
A word $w\in\Sigma^{*}$ is a central word if and only if it satisfies one of the following equations, for some $u,v\in\Sigma^{*}$ and $n\geq 0$:
\begin{enumerate}
\item $w=a^{n},$
\item $w=b^{n},$
\item $w=uabv=vbau.$
\end{enumerate}
Moreover, if case $3$ holds one has that $u$ and $v$ are central words.
\end{proposition}

It is known that if $w=uabv=vbau$ is a central word, then the longest between $u$ and $v$ has not internal occurrences in $w$. Thus, $w$ is a closed word. Note that this is also a consequence of our Theorem \ref{theor:palclo}.

\begin{proposition}[\cite{Del97}]\label{prop:newcent}
 A word $w\in\Sigma^{*}$ is a central word if and only if $w$ is a palindrome and $awa$ and $bwb$ are both Sturmian words.
\end{proposition}

\begin{proposition}[\cite{DelMi94}]\label{prop:centrbis}
  A word $w\in\Sigma^{*}$ is a central word if and only if there exists a Sturmian word $w'$ such that $w$ is a bispecial factor of $w'$.
\end{proposition}

Recall that a word $w$ is not balanced if and only if it contains a \emph{pathological pair} $(f,g)$, that is, $f$ and $g$ are factors of $w$ of the same length but they do not verify \eqref{eq:bal}. Moreover, if $f$ and $g$ are chosen of minimal length, there exists a palindrome $u$ such that $f=aua$ and $g=bub$, for two different letters $a$ and $b$ (see \cite[Proposition 2.1.3]{LothaireAlg} or \cite[Lemma 3.06]{CoHe73}). 

We can also state that $f$ and $g$ are Sturmian words, since otherwise they would contain a  pathological pair shorter than $|f|=|g|$ and hence $w$ would contain such a pathological pair, against the minimality of $f$ and $g$. So the word $u$ is a palindrome such that $aua$ and $bub$ are Sturmian words, that is, by Proposition \ref{prop:newcent}, $u$ is a central word.

Moreover, we have the following

\begin{lemma}
 Let $w$ be a binary non-Sturmian word and $(f,g)$ a pathological pair of $w$ of minimal length. Then there exists a unique central word $u$ such that $f=aua$ and $g=bub$.
\end{lemma}

\begin{proof}
 Suppose by contradiction that there exists a word $u'\neq u$, with $|u|=|u'|$, such that $w$ contains as factors $aua$, $bub$, $au'a$ and $bu'b$. Let $h$ be the longest common prefix between $u$ and $u'$. Since $u\neq u'$ there exist distinct letters $x,y\in \Sigma$ such that $hx$ is a prefix of $u$ and $hy$ is a prefix of $u'$. Then $w$ would contain as factors $aha$, $ahb$, $bha$ and $bhb$. Therefore, $w$ would contain a pathological pair $(aha,bhb)$ shorter than $(f,g)$, a contradiction. 
\end{proof}

The following lemma is attributed to Aldo de Luca in \cite{Dal02}. We provide a proof for the sake of completeness.

\begin{lemma}\label{lem:separation}
 Let $w$ be a non-Sturmian word and $(f,g)$ the pathological pair of $w$ of minimal length. Then $f$ and $g$ do not overlap in $w$.
\end{lemma}

\begin{proof}
By contradiction, suppose that $w$ contains a factor $uvz$ such that $f=uv$  and $g=vz$, with $|v|>0$. Since $||f|_{a}-|g|_{a}|=||u|_{a}-|z|_{a}|$, then $(u,z)$ would be a pathological pair of $w$ shorter than $(f,g)$. 
\end{proof}

We then have:

\begin{proposition}\label{prop:centralroot}
  Let $w$ be a binary non-Sturmian word. Then there exists a unique central word of minimal length $u$ such that $w=v_{1}auav_{2}bubv_{3}$, for $a,b\in \Sigma$ distinct letters and $v_{1},v_{2},v_{3}\in \Sigma^{*}$. We call the word $u$ the \emph{central root} of $w$.
\end{proposition}


The following is the characterization of trapezoidal non-Sturmian words given by D'Alessandro \cite{Dal02}.

\begin{theorem}\label{theor:dal}
 Let $w$ be a binary non-Sturmian word. Then w is trapezoidal if and only if
\begin{equation}\label{eq:dal}
 w = pq, \quad\text{ with }p \in \Suff(\{\tilde{z}_{f}\}^{*}),\quad q\in \Pref(\{z_{g}\}^{*})
\end{equation}
where $(f,g)$ is the pathological pair of $w$ of minimal length. In particular, if $w$ is trapezoidal, then $K_{w}=|q|$ and the longest right special factor of $w$ is the prefix of $w$ of length $R_{w}-1$, that is, the longest proper prefix of $p$.
\end{theorem}

Hence, trapezoidal words are either Sturmian or of the form described in Theorem \ref{theor:dal}.


\begin{example}
 Let $w=aaababa$ be the non-Sturmian trapezoidal word considered in Example \ref{ex:1}. We have  $f=aaa$ and $g=bab$, so that $\tilde{z}_{f}=a$ and $z_{g}=ba$. The word $w$ can be written as $w=pq$, with $p=aaa$ and $q=baba$. 
\end{example}

\begin{remark}
 From now on, according to Theorem \ref{theor:dal}, $(f,g)$ is the pair of pathological factor of $w$ of minimal length.
\end{remark}

We now show that the words $p$ and $q$ of the factorization \eqref{eq:dal} are Sturmian words. We premise the following known result:

\begin{proposition}[cf.~\cite{DelDel06}]
Let $w$ be a finite word and $z_{w}$ its fractional root. Then $w$ is Sturmian if and only if so is every power of $z_{w}$.
\end{proposition}

\begin{lemma}\label{lem:pk}
Let $w$ be a trapezoidal non-Sturmian word. Then the words  $p$ and $q$ of the factorization \eqref{eq:dal} are Sturmian words. 
\end{lemma}

\begin{proof}
Since $f$ and $g$ are Sturmian, by the preceding proposition all powers of
$z_{f}$ and $z_{g}$ are Sturmian as well. Observe that
$p\in\Suff(\{\tilde z_{f}\}^{*})$ can be restated as
$\tilde p\in\Pref(\{z_{f}\}^{*})$. As prefixes and reversals of Sturmian words are Sturmian, the assertion follows.
\end{proof}


The following result, due to de Luca, Glen and Zamboni \cite{DelGlZa08}, states that trapezoidal palindromes are all Sturmian. We give a new proof based on Theorem \ref{theor:dal}.

\begin{theorem}\label{theor:trappal}
 The following conditions are equivalent:
 
\begin{enumerate}
 \item $w$ is a trapezoidal palindrome;
 \item $w$ is a Sturmian palindrome.
\end{enumerate}
\end{theorem}

\begin{proof} 
Let $w$ be a trapezoidal palindrome. If $w$ is non-Sturmian we can write, by Theorem \ref{theor:dal}, $w=pq=v_{1}fgv_{2}=\tilde{v}_{2}gf\tilde{v}_{1}$, with $v_{1},v_{2}\in \Sigma^{*}$ such that $p=v_{1}f$ and $q=gv_{2}$. Without loss of generality, we can assume $|v_1|\geq |v_2|$; then $p=v_1f$ contains $f$ and $g$ as factors, a contradiction since, by Lemma \ref{lem:pk}, $p$ is Sturmian. So $w$ cannot be a non-Sturmian word. 

Hence, we proved that trapezoidal palindromes are Sturmian. Since by Proposition \ref{prop:sturmtrap}, Sturmian words are trapezoidal, the assertion follows.
%
\end{proof}

%
%
%

\section{An enumerative formula for trapezoidal words}\label{sec:formula}

Mignosi \cite{Mig91} proved the following formula (see also \cite{Lip82}) for the number $S(n)$ of Sturmian words of length $n$: 
\begin{equation}\label{eq:sturm}
 S(n)=1+\sum_{i=1}^{n}(n-i+1)\phi(i),
\end{equation}
where $\phi$ is the Euler totient function.

Therefore, by Proposition \ref{prop:sturmtrap}, in order to count the number of trapezoidal words of length $n$, it is sufficient to count the number of non-Sturmian trapezoidal words of that length. Recall from Theorem \ref{theor:dal} that every non-Sturmian trapezoidal word $w$ can be uniquely written as $w=pq$, with $p \in \Suff(\{\tilde{z}_{f}\}^{*})$ and $q\in \Pref(\{z_{g}\}^{*})$, with $f=aua$, $g=bub$, where $u$ is the central root of $w$ and $(f,g)$ is the pathological pair of $w$ of minimal length. Our goal is to count, for a given central word $u$, how many non-Sturmian trapezoidal words exist with central root $u$ and length $n$.

Let us premise the following technical lemma, which is a straightforward consequence of Theorem \ref{theor:dal}.

\begin{lemma}\label{lem:tech}
 Let $u$ be a central word, and $x,y\in \Sigma$ different letters. The non-Sturmian trapezoidal words having central root $u$ are exactly the words of the form $pq$, for any $p\in \Suff(\{\tilde{z}_{xux}\}^{*})$ with $|p|\ge |u|+2$, and $q\in \Pref(\{z_{yuy}\}^{*})$ with $|q|\ge |u|+2$.
\end{lemma}

\begin{theorem}\label{theor:enum}
 For all $n\geq 0$, the number of trapezoidal words of length $n$ is
\[1+\sum_{i=1}^{n}(n-i+1)\phi(i) \ +  \sum_{i=0}^{\lfloor (n-4)/2\rfloor}2(n-2i-3)\phi(i+2).\]
\end{theorem}

\begin{proof}
We first notice that for any word $z$ and any $j \ge 0$, there exists a unique word of length $j$ in $\Suff(\{\tilde{z}\}^{*})$ (resp.~in $\Pref(\{z\}^{*})$). From this remark and Lemma \ref{lem:tech}, we deduce that for any $n\ge 4$ and for any central word $u$ of length $0\le |u|\le \lfloor (n-4)/2\rfloor$, there are exactly $2(n-2(|u|+2)+1)=2(n-2|u|-3)$ non-Sturmian trapezoidal words of length $n$ having central root $u$. Since there are $\phi(i+2)$ central words of length $i$ (see \cite{DelMi94}), we have that the number of non-Sturmian trapezoidal words of length $n$ is  
\begin{equation}\label{eq:trapform}
 T(n)=\sum_{i=0}^{\lfloor (n-4)/2\rfloor}2(n-2i-3)\phi(i+2).
\end{equation}
Since the number of trapezoidal words of length $n$ is $S(n)+T(n)$, by \eqref{eq:sturm} and \eqref{eq:trapform} the statement follows.
\end{proof}

In Table \ref{table:trap} we give the number of trapezoidal words of each length up to 20.

\begin{table}[ht]
\centering  
\begin{small}
\begin{raggedright}
\begin{tabular}{c *{30}{@{\hspace{2.7mm}}l}}
 $n$    & 1\hspace{1ex} & 2\hspace{1ex} & 3\hspace{1ex} & 4\hspace{1ex} & 5\hspace{1ex} & 6\hspace{1ex} & 7\hspace{1ex} &
8\hspace{1ex} & 9\hspace{1ex} & 10 & 11 & 12 & 13 & 14 & 15 & 16 & 17 & 18 & 19 & 20\\
\hline \rule[-6pt]{0pt}{22pt}
$\textit{Trap} \cap \Sigma^{n}$ & 2 & 4 & 8 & 16 & 28 & 46 & 70 & 102 & 140 & 190 & 250 & 318 & 398 & 496 & 602 & 724 & 862 & 1018 & 1192 & 1382 \\
\hline \rule[-6pt]{0pt}{22pt}
\end{tabular}
\end{raggedright}\caption{The number of trapezoidal words of each length up to 20.\label{table:trap}}
\end{small}
\end{table}

\section{Open and closed trapezoidal words}\label{sec:openclosed}

In this section we derive some properties of open and closed trapezoidal words. Recall that a word $w$ is closed if it is empty or it has a factor (different from $w$) occurring exactly twice in $w$, as a prefix and as a suffix; otherwise it is open.

\begin{remark}
Let $w$ be a non-empty word over $\Sigma$. The following characterizations of closed words follow easily from the definition:
 
\begin{enumerate}
 \item $w$ has a factor occurring exactly twice in $w$, as a prefix and as a suffix of $w$; 
 \item the longest repeated prefix of $w$ does not have internal occurrences in $w$, that is, occurs in $w$ only as a prefix and as a suffix;
  \item the longest repeated suffix of $w$ does not have internal occurrences in $w$, that is, occurs in $w$ only as a suffix and as a prefix;
  \item the longest repeated prefix of $w$ is not right special in $w$;
    \item the longest repeated suffix of $w$ is not left special in $w$;
 \item $w$ has a border that does not have internal occurrences in $w$;
 \item the longest border of $w$ does not have internal occurrences in $w$;
 \item $w$ is the complete return to its longest repeated prefix;
 \item $w$ is the complete return to its longest border;
 \item $w=uv=zu$, with $v,z$ non-empty, and $\Fact(w)\cap \Sigma u \Sigma = \emptyset$.
\end{enumerate}
\end{remark}

\begin{proposition}\label{prop:rev}
Let $w$ be a word. Then $w$ is open (resp.~closed) if and only if $\tilde{w}$ is open (resp.~closed).
\end{proposition}

\begin{proof}
Straightforward.
%
\end{proof}

The following proposition gives a characterization of open trapezoidal words.

\begin{proposition}\label{prop:openspe}
 Let $w$ be a trapezoidal word. Then the following conditions are equivalent:
\begin{enumerate}
\item $w$ is open;
 \item the longest repeated prefix of $w$ is also the longest right special factor of $w$;
 \item the longest repeated suffix of $w$ is also the longest left special factor of $w$.
\end{enumerate}
\end{proposition}

\begin{proof}
$(1)\Rightarrow (2)$. Let $h_{w}$ be the longest repeated prefix of $w$ and $x$ the letter such that $h_{w}x$ is a prefix of $w$. Since $w$ is not closed, $h_{w}$ has a second non-suffix occurrence in $w$ followed by a letter $y$. Since $h_{w}$ is the longest repeated prefix of $w$, we have $y\neq x$. Therefore, $h_{w}$ is right special in $w$.
 
 Suppose that $w$ has a right special factor $r$ longer than $h_{w}$. Since $w$ is trapezoidal, $w$ has at most one right special factor for each length (Proposition \ref{prop:trap}). Since the suffixes of a right special factor are right special factors, we have that $h_{w}$ must be a proper suffix of $r$. Since $r$ is right special in $w$, it has at least two occurrences in $w$ followed by different letters. This implies a non-prefix occurrence of $h_{w}x$ in $w$, against the definition of $h_{w}$.
 
%


$(2)\Rightarrow (1)$. Let $h_{w}$ be the longest repeated prefix of $w$. Since $h_{w}$ is a right special factor of $w$ there exist different letters $a$ and $b$ such that $h_{w}a$ and $h_{w}b$ are factors of $w$. This implies that $h_{w}$ cannot appear in $w$ only as a prefix and as a suffix, that is, $w$ cannot be a closed word. 

$(1)\Leftrightarrow (3)$ can be proven symmetrically.
\end{proof}

\begin{proposition}\label{prop:HRKL}
 Let $w$ be a trapezoidal word. If $w$ is open, then $H_{w}=R_{w}$ and $K_{w}=L_{w}$. If $w$ is closed, then $H_{w}=K_{w}$ and $L_{w}=R_{w}$.
\end{proposition}

\begin{proof}
The claim for open trapezoidal words follows directly from Proposition \ref{prop:openspe}. 
 
Suppose that $w$ is a closed trapezoidal word. Then $H_{w}=K_{w}$ (since $w$ is closed) and therefore $L_{w}=R_{w}$ (since $w$ is trapezoidal).
\end{proof}

The latter proposition, however, does not characterize open and closed trapezoidal words. Actually, one can have  $H_{w}=K_{w}=L_{w}=R_{w}$, as is the case in the word $abba$, which is closed, and also in the word $aaba$, which is open.

Open trapezoidal words can be Sturmian (e.g.~$aaabaa$) or not (e.g.~$aaabab$). Closed trapezoidal words, instead, are always Sturmian, as shown in the following proposition, that can also be found in a paper of the first two authors together with Aldo de Luca (cf.~\cite{BuDelDel09}, Proposition 3.6). Here we report a proof following a totally different approach, based on Theorem \ref{theor:dal}.

\begin{proposition}\label{prop:cloStur}
Let $w$ be a trapezoidal word.  If $w$ is closed, then $w$ is Sturmian.
\end{proposition}

\begin{proof}
 Suppose that $w$ is not Sturmian. Then, by Theorem \ref{theor:dal}, $w=pq$, $p \in \Suff(\{\tilde{z}_{f}\}^{*})$, $q\in \Pref(\{z_{g}\}^{*})$, with $(f,g)$ being the pathological pair of $w$ of minimal length, $K_{w}=|q|$ (and hence $R_{w}=|p|$ since $w$ is trapezoidal) and the longest right special factor of $w$ is the prefix of $w$ of length $R_{w}-1$. 
 
Let $k_{w}$ be the suffix of $q$ of length $K_{w}-1$, that is, the longest repeated suffix of $w$. Since $w$ is closed, the only other occurrence of $k_{w}$ in $w$ is as a prefix of $w$.
 
  
If $R_{w}\ge K_{w}$, then $k_{w}$ is a prefix of the longest right special factor of $w$. This is a contradiction with the fact that $k_{w}$ does not have internal occurrences in $w$.  
  
If $R_{w}<K_{w}$, then $p$ is a prefix of $k$ and hence $p$ is a factor of $q$. This implies that $f$ is a factor of $q$, and therefore $q$ contains both $f$ and $g$ as factors. Hence $q$ would be non-Sturmian, contradicting Lemma \ref{lem:pk}.
\end{proof}

As a corollary of Proposition \ref{prop:cloStur}, we have that every trapezoidal word is open or Sturmian. 
%
%
%
We now focus on closed trapezoidal words and their special factors.

\begin{lemma}\label{lem:clo1}
Let $w$ be a closed trapezoidal word and $u$ the longest left special factor of $w$. Then $u$ is also the longest right special factor of $w$ (and thus $u$ is a bispecial factor of $w$). Moreover, $u$ is a central word.  
\end{lemma}

\begin{proof}
Let $u$ be the longest left special factor of $w$, so that there exist different letters $a,b\in \Sigma$ such that $au$ and $bu$ are factors of $w$. 

We claim that $au$ and $bu$ both occur in $w$ followed by some letter. Indeed, suppose the contrary. Then one between $au$ and $bu$, say $au$, appears in $w$ only as a suffix. Let $k_{w}$ be the longest repeated suffix of $w$. Since $u$ is a repeated suffix of $w$, we have $|k_{w}|\ge |u|$. If $|k_{w}|=|u|$, then $k_{w}=u$ and since $w$ is closed, $u$ appears in $w$ only as a prefix and as a suffix, against the hypothesis that $bu$ is a factor of $w$. So $|k_{w}|>|u|$ and therefore $au$ must be a suffix of $k_{w}$. This implies an internal occurrence of $au$ in $w$, a contradiction.

So, there exist letters $x,y$ such that $aux$ and $buy$ are factors of $w$. Now, we must have $x\neq y$, since otherwise $ux$ would be a left special factor of $w$ longer than $u$. Thus $u$ is right special in $w$. Since $w$ is closed, we have, by Proposition \ref{prop:HRKL}, $L_{w}=R_{w}$, and therefore $u$ is the longest right special factor of $w$.

By Proposition \ref{prop:cloStur}, $w$ is Sturmian. Since $u$ is a bispecial factor of a Sturmian word, we have by Proposition \ref{prop:centrbis} that $u$ is a central word.
\end{proof}

\begin{example}
Let $w=aababaaba$. The longest repeated prefix of $w$ is $aaba$, which is also the longest repeated suffix and does not have internal occurrences, hence $w$ is a closed trapezoidal word.

The longest left special factor of $w$ is $aba$, which is also the longest right special factor of $w$ and it is a central word.
\end{example}

By Theorem \ref{theor:trappal}, trapezoidal palindromes coincide with Sturmian palindromes, so the next results will be stated in terms of Sturmian palindromes but one can replace the word ``Sturmian'' with the word ``trapezoidal''. 

Some results on Sturmian palindromes can be found in \cite{DelDel06a}. In particular, we recall here the following one.

\begin{theorem}[\cite{DelDel06a}, Theorem 29]\label{theor:delucas}
 A palindrome $w\in \Sigma^{*}$ is Sturmian if and only if its minimal period $\pi(w)$ satisfies $\pi_{w}=R_{w}+1$.
\end{theorem}

The next theorem shows that Sturmian palindromes, and so, by Theorem \ref{theor:trappal}, trapezoidal palindromes, are all closed words. 

\begin{theorem}\label{theor:palclo}
 Let $w$ be a Sturmian palindrome. Then $w$ is closed.
\end{theorem}

\begin{proof}
By contradiction, suppose that $w$ is open. Then $h_{w}$, the longest repeated prefix of $w$,  is also the longest right special factor of $w$ (Proposition \ref{prop:openspe}). Since $w$ is a palindrome, we have that the longest repeated suffix of $w$ is $\tilde{h}_{w}$, the reversal of $h_{w}$. In particular, then, $K_{w}=H_{w}$.

By Proposition \ref{prop:HRKL}, $H_{w}=R_{w}$ and $K_{w}=L_{w}$. Thus we have $H_{w}=R_{w}=K_{w}=L_{w}=|w|/2$, since $w$ is trapezoidal (see Proposition \ref{prop:trap}). It follows that $w=h_{w}xx\tilde{h}_{w}$, for a letter $x\in \Sigma$. By Theorem \ref{theor:delucas}, the period of $w$ is $R_{w}+1=|h_{w}xx|$, so that $\tilde{h}_{w}=h_{w}$.  Therefore, we have $w=h_{w}xxh_{w}$. 

Since $h_{w}$ is right special in $w$, there exists a letter $y\neq x$ such that $h_{w}y$ is a factor of $w$. Now, any occurrence of $h_{w}y$ in $w$ cannot be preceded by the letter $x$, since $h_{w}=\tilde{h}_{w}$ is the longest repeated suffix of $w$. Thus $w$ contains the factor $yh_{w}y$. 

Hence, $w$ contains both $h_{w}xx$ and $yh_{w}y$ as factors, and this contradicts the fact that $w$ is Sturmian.
\end{proof}

\begin{figure}[ht]
\begin{center}
\includegraphics[height=7cm]{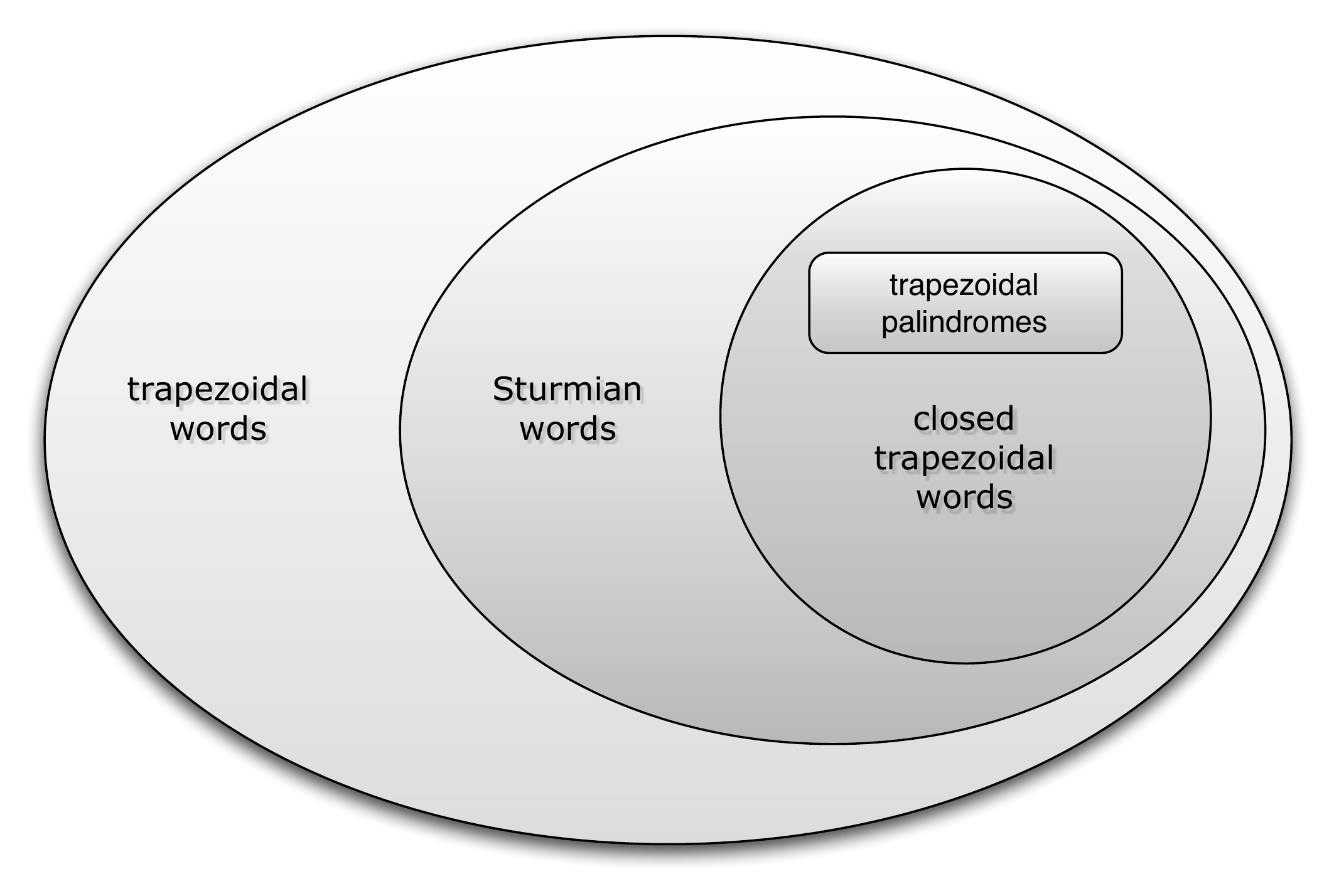}
\caption{A Venn diagram illustrating the inclusion relations of trapezoidal words.}
\label{fig:insiemi}
\end{center}
\end{figure}

From Theorem \ref{theor:palclo} and Lemma \ref{lem:clo1}, we derive the following characterization of the special factors of Sturmian palindromes.

\begin{corollary}\label{cor:sturpal}
 Let $w$ be a Sturmian palindrome. Then the longest left special factor of $w$ is also the longest right special factor of $w$ and it is a central word.
\end{corollary}

Let us recall the following characterization of rich words from~\cite{GlJuWiZa09}:

\begin{proposition}\label{prop:compret}
A word $w$ is rich if and only if every factor of $w$ which is a complete return to a palindrome is palindromic itself.
\end{proposition}
We have the following result:

\begin{proposition}\label{prop:palh}
 Let $w$ be a closed trapezoidal word and $h_{w}$ its longest repeated prefix. Then $w$ is a Sturmian palindrome if and only if $h_{w}$ is a palindrome.
\end{proposition}

\begin{proof}
The word $w$ is a complete return to $h_{w}$, and it is rich by Proposition~\ref{prop:traprich}. Since trapezoidal palindromes coincide with Sturmian palindromes by Theorem \ref{theor:trappal}, the preceding proposition implies the assertion.
\end{proof}

The general structure of a closed trapezoidal word is depicted in \figurename~\ref{fig:clotrap}.

\begin{figure}[ht]
\begin{center}
\includegraphics[height=1cm]{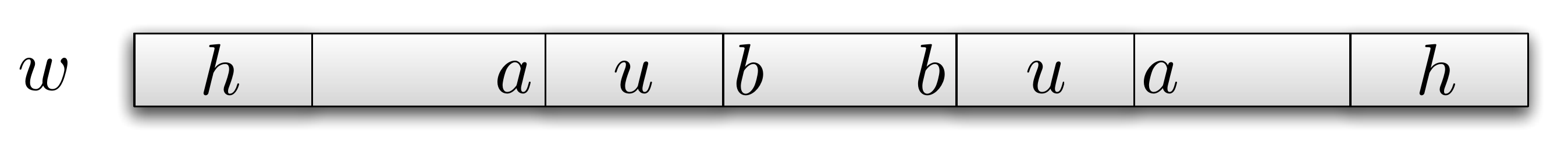}
\caption{The structure of a closed trapezoidal word. The longest repeated prefix, $h_{w}$, is also the longest repeated suffix and does not have internal occurrences. The longest left special factor, $u$, is also the the longest right special factor and it is a central word (Lemma \ref{lem:clo1}). The word is a palindrome if and only if $h_{w}$ is a palindrome (Proposition \ref{prop:palh}).}
\label{fig:clotrap}
\end{center}
\end{figure}

\begin{remark}\label{rem:richpal}
Using Proposition \ref{prop:compret}, one can prove that Theorem \ref{theor:palclo} can be generalized to rich palindromes. Indeed, let $w$ be a rich palindrome, and $u$ the longest border of $w$. In order to prove that $w$ is closed, it is sufficient to prove that $u$ does not have internal occurrences in $w$. Suppose the contrary. Then $w$ would contain a proper prefix of the form $uzu$. By Proposition \ref{prop:compret}, $uzu$ is a palindrome. Hence, $uzu$ would be a suffix of $w$, against the maximality of $u$.

Notice that the converse is not true. The word $w=aaababbaabbabaaa$ is a closed palindrome of length $16$ but contains only $16$ palindromic factors, so it is not rich. 
\end{remark}

\section{Semicentral words}\label{sec:sesqui}


In what follows, we let by $h_{w},k_{w},r_{w},l_{w}$ respectively denote the longest repeated prefix, the longest repeated suffix, the longest right special factor and the longest left special factor of the word $w$. An intriguing class of trapezoidal words is that of words in which all these factors coincide.

\begin{definition}
 Let $w$ be a trapezoidal word. We say that $w$ is \emph{semicentral} if $h_{w}=k_{w}=r_{w}=l_{w}$.
\end{definition}

Notice that, by Proposition \ref{prop:openspe}, semicentral words can also be defined as those open trapezoidal words in which the longest repeated prefix is also the longest repeated suffix.

The following theorem provides a characterization of semicentral words.

\begin{theorem}\label{theor:sesqui}
A trapezoidal word $w$ is semicentral if and only if there exists a central word $u$ such that $w=uxyu$ with $x,y\in\Sigma$, $x\neq y$.
\end{theorem}

\begin{proof}
By~\eqref{eq:trap}, if $w$ is semicentral we can write $w=uxyu$, with
$u=h_{w}=k_{w}=r_{w}=l_{w}$ and $x,y$ (possibly equal) letters in $\Sigma$. Since $u$ is a bispecial factor of $w$, it must have a third occurrence in $w$; as $u=h_{w}=k_{w}$, such an occurrence cannot be preceded by $y$ or followed by $x$. We have then
\begin{equation}
\label{eq:sesq}
w=uxyu=\alpha\bar yu\bar x\beta
\end{equation}
where $\Sigma=\{x,\bar x\}=\{y,\bar y\}$, $\alpha,\beta\in\Sigma^{*}$ and
$|\alpha|+|\beta|=|u|$.
If $\alpha=u$, then $uxyu=u\bar yu\bar x$, so that $x\neq y$ and $yu=uy$. Thus $u=y^{n}$ for some $n\geq 0$, so that it is a central word. Similarly, if 
$\beta=u$ it follows that $x\neq y$ and $u=x^{n}$ is central.

Let then $0<|\alpha|,|\beta|<|u|$. From~\eqref{eq:sesq} it follows that
there exists a prefix $u_{1}x$ of $u$ such that $ux=\alpha\bar yu_{1}x$, and a suffix $yu_{2}$ of $u$ such that $yu=yu_{2}\bar x\beta$.
Hence, as $|u|=|\alpha|+|\beta|$, we obtain
\[u=u_{1}xyu_{2}=u_{2}\bar x\bar yu_{1}\;.\]
Again, this implies $x\neq y$, otherwise the number of occurrences of $x$ in $u_{1}xyu_{2}$ and $u_{2}\bar x\bar yu_{1}$ would be different. Thus, $u=u_{1}xyu_{2}=u_{2}yxu_{1}$ is central by
Proposition~\ref{prop:ipcent}.

Conversely, let $u$ be a central word and $w=uxyu$, with $\Sigma=\{x,y\}$. We claim that $xuy$ occurs in $w$. This is clear if $u$ is a power of $x$ or $y$; by Proposition~\ref{prop:ipcent}, the other possibility is $u=u_{1}xyu_{2}=u_{2}yxu_{1}$ for some words $u_{1},u_{2}$, so that
\[w=uxyu=u_{2}y\cdot xu_{1}xyu_{2}y\cdot xu_{1}\]
gives the desired occurrence of $xuy$. Hence $u$ is a right special factor and a repeated suffix of $w$, so that $R_{w}\ge |u|+1$ and $K_{w}\ge |u|+1$. Since $|w|=2(|u|+1)\le R_{w}+K_{w}$, it follows from \eqref{eq:RK} that $w$ is trapezoidal and $k_{w}=r_{w}=u$. Analogously, since $u$ is a left special factor and a repeated prefix of $w$, one proves that $h_{w}=l_{w}=u$.
\end{proof}

By Proposition \ref{prop:ipcent}, every central word $u$ different from a power of a single letter can be written as $u=u_{1}xyu_{2}=u_{2}yxu_{1}$ for central words $u_{1}$ and $u_{2}$ and distinct letters $x,y$. Theorem \ref{theor:sesqui} provides a similar characterization for semicentral words (and this motivates the choice of the name).

It is also worth noticing that semicentral words are non-strictly bispecial Sturmian words (see \cite{Fi12}), that is, for any semicentral word $w$ one has that $aw,bw,wa$ and $wb$ are all Sturmian words, but nevertheless either $awb$ or $bwa$ is not Sturmian.

\begin{example}
 The semicentral words of length $8$ are: $aaaabaaa$, $aaabaaaa$, $abaababa$, $ababaaba$, $bababbab$, $babbabab$, $bbbabbbb$ and $bbbbabbb$.
\end{example}

As a consequence of Theorem \ref{theor:sesqui}, we have that the number $SC(n)$ of semicentral words of length $n$ is:

\[SC(n) = \left\{ \begin{array}{lllll}
0, & \text{ if $n$ is odd,}\\
2\phi\left(\frac{n}{2}+1\right), & \text{ if $n$ is even.}
\end{array} \right.\]


\section{Open and closed prefixes of the Fibonacci word}\label{sec:fibo}

In this section, we give an interesting example of the open/closed dichotomy in trapezoidal words, by analyzing prefixes of the Fibonacci infinite word under this new light. The Fibonacci infinite word $f$ is the word \[f=abaababaabaababaababaabaaba\cdots\] obtained as the limit of the substitution $a\mapsto ab$, $b \mapsto a$. 
Note that every prefix of the Fibonacci word is Sturmian, and therefore trapezoidal.

We investigate which prefixes of $f$ are open and which are closed. If one writes up the list of the prefixes of $f$ for each length (see Table \ref{tab:fibo}), can notice that the sequences of ``open'' and ``closed'' alternate following the Fibonacci sequence $F_{i}$ (defined by $F_{1}=F_{2}=1$ and $F_{i+2}=F_{i+1}+F_{i}$  for every $i\ge 1$). More precisely, the numbers of consecutive prefixes that are closed (resp.~open), that is, the lengths of the runs of ``c'' and ``o'' in  Table \ref{tab:fibo}, are the Fibonacci numbers. In Theorem \ref{theor:fibo} we show this fact. 

\begin{table}[ht]
\centering  
\begin{small}
\begin{tabular}{r*{9}{@{\hspace{1.5em}}c}*{11}{c}}
 $n$    & 1 & 2 & 3 & 4 & 5 & 6 & 7 &
8 & 9 & 10 & 11 & 12 & 13 & 14 & 15 & 16 & 17 & 18 & 19 & 20\\
\hline \\[-1ex]
$f_{1\cdots n}$ & c & o & c & o & c & c & o & o & c & c & c & o & o & o & c & c & c & c & c & o\\[0.5ex]
\hline 
\end{tabular}
\label{tab:fibo}
\end{small}
\caption{Which prefixes of the Fibonacci word are closed and which are open.}
\end{table}

%

\begin{theorem}\label{theor:fibo}
Let $w$ be a prefix of the Fibonacci word $f$. Then $w$ is open if and only if there exists $i$ such that $F_{i+1}-1\le |w| \le 2F_{i}-2$.
\end{theorem}

\begin{proof}
First, let us recall some well known properties of the Fibonacci word. Let $S=\{\epsilon,a,aba,abaaba,abaababaaba,\ldots\}$ be the set of palindromic prefixes of $f$. It is well known that the elements of $S$ are the prefixes $s_{i}$ of $f$ such that $|s_{i}|=F_{i}-2$, $i\ge 3$.  The words in $S$ are all central words. 
Also, for every $i\geq 4$ we have
\begin{equation}
\label{eq:xiyi}
s_{i+1}=s_{i}x_iy_is_{i-1}=s_{i-1}y_ix_is_{i},
\end{equation}
where $x_i=b$, $y_i=a$ for even $i$, and $x_i=a$, $y_i=b$ for odd $i$.

Our assertion is trivially verified if $|w|\leq 3=F_{5}-2$. Hence we only need to show that for all $i\geq 4$, a prefix $w$ of $f$ with
$F_{i+1}-1\leq |w|<F_{i+2}-1$ is
\begin{enumerate}
\item open if $F_{i+1}-1\leq |w|\leq 2F_{i}-2$, and
\item closed otherwise, i.e., when $2F_{i}-1\leq |w|\leq F_{i+2}-2$.
\end{enumerate}

Let us fix $i\geq 4$ and set $x=x_i$, $y=y_i$ as in~\eqref{eq:xiyi}. We know that the prefix of $f$ of length $F_{i+1}-2$ is $s_{i+1}$. By Theorem~\ref{theor:palclo}, it is a closed word and its longest repeated prefix is
$s_{i}$. Moreover, $s_{i}$ is also the longest repeated prefix of the semicentral word $s_{i}xys_{i}$, that is the prefix of $f$ of length $2F_{i}-2$. Therefore, $s_{i}$ is the longest repeated prefix for all $w$'s in between, i.e., for each $w\in\Pref(f)$ with $F_{i+1}-1\leq |w|\leq 2F_{i}-2$. As $s_{i}$ has an internal occurrence in all such prefixes, they are all open, proving point 1.

It remains to show that any prefix of
\[s_{i+2}=s_{i}xys_{i+1}=s_{i}xys_{i}xys_{i-1}\]
longer than $s_{i}xys_{i}$ is closed. Indeed, such words have a border which is strictly longer than $s_{i}$, and which cannot have any internal occurrences for otherwise $s_{i}x$ would have a second occurrence in $s_{i}xys_{i}$. 
\end{proof}

\begin{table}[h]
\begin{center}
  \begin{tabular}{| l | l | c | l | c |}
  
    prefix of $f$  & length &  open/closed & example  \\    \hline 
     $s_{i}xys_{i-1}=s_{i+1}$     &   $F_{i+1}-2$       & closed   & $abaababaaba$        \\
     $s_{i}xys_{i-1}y$    &   $F_{i+1}-1$       & open   &    $abaababaabaa$      \\   
     $s_{i}xys_{i-1}yx$   &   $F_{i+1}$         & open   & $abaababaabaab$        \\
     \ldots      &   \ldots          & \ldots   &  \ldots        \\
	 $s_{i}xys_{i}$     &   $2F_{i}-2$       & open      & $abaababaabaaba$       \\
	 $s_{i}xys_{i}x$     &   $2F_{i}-1$       & closed   & $abaababaabaabab$          \\
	 $s_{i}xys_{i}xy$     &   $2F_{i}$       & closed      & $abaababaabaababa$       \\
	 \ldots      &   \ldots          & \ldots      &  \ldots      \\
	 $s_{i}xys_{i+1}=s_{i+2}$   &   $F_{i+2}-2$     & closed   & $abaababaabaababaaba$ \\
	 $s_{i}xys_{i+1}x$   &   $F_{i+2}-1$     & open   & $abaababaabaababaabab$ \\
	 
    \hline
  \end{tabular}
\end{center}\caption{The structure of the prefixes of $f$ with respect to the palindromic prefixes $s_{i}$.}\label{tab:example}
\end{table}

\begin{corollary}
 The longest word in a run of closed prefixes of $f$ is a central word. The longest word in a run of open prefixes of $f$ is a semicentral word.
\end{corollary}

\section{Conclusions and open problems}\label{sec:conc}

In this paper we have investigated trapezoidal words, which are a natural generalization of finite Sturmian words. We exhibited a formula for counting trapezoidal words and studied open and closed trapezoidal words separately. 

Trapezoidal words form a subclass of the class of rich words, that are words containing the maximum number of palindromic factors. A challenging problem is that of finding an enumerative formula for rich words. This problem is still open, even in the binary case. We think that separating rich words in open and closed could give further insights on these words.

More generally, we think that the open/closed dichotomy can be useful in the study of other classes of words. For instance, it would be interesting to generalize the results obtained in Section~\ref{sec:fibo} for the Fibonacci word to the class of standard words.


\end{document}